\documentclass[11pt]{article}
\usepackage{graphicx}
\usepackage{multirow}

\usepackage[T1]{fontenc}
\usepackage[sc]{mathpazo}
\usepackage{amsmath}
\usepackage{amssymb}
\usepackage{enumerate}
\usepackage{amsthm}
\usepackage{amsfonts,mathrsfs}

\usepackage{hyperref}
\hypersetup{pdfpagemode=UseNone}

\newcommand{\tinyspace}{\mspace{1mu}}

\newcommand{\abs}[1]{\left\lvert\tinyspace #1 \tinyspace\right\rvert}

\newcommand{\norm}[1]{\left\lVert\tinyspace #1 \tinyspace\right\rVert}

\newcommand{\setft}[1]{\mathrm{#1}}

\newcommand{\density}[1]{\setft{D}\left(#1\right)}

\def\I{\mathbb{1}}

\newenvironment{mylist}[1]{\begin{list}{}{
    \setlength{\leftmargin}{#1}
    \setlength{\rightmargin}{0mm}
    \setlength{\labelsep}{2mm}
    \setlength{\labelwidth}{8mm}
    \setlength{\itemsep}{0mm}}}
    {\end{list}}

\newcommand{\iinner}[2]{\langle #1 | #2\rangle}
\newcommand{\out}[2]{| #1\rangle\langle #2 |}

\newcommand{\Innerm}[3]{\left\langle #1 \left| #2 \right| #3 \right\rangle}

\newcommand{\Pa}[1]{\left(#1\right)}

\newcommand{\set}[1]{\{#1\}}

\newcommand{\ket}[1]{|#1\rangle}

\DeclareMathOperator{\trace}{Tr}

\newcommand{\Ptr}[2]{\trace_{#1}\Pa{#2}}

\newcommand{\Tr}[1]{\Ptr{}{#1}}

\def\cH{\mathcal{H}}

\def\rB{\mathrm{B}}
\def\rF{\mathrm{F}}

\def\rS{\mathrm{S}}

\newtheorem{thrm}{Theorem}[section]

\newtheorem{cor}[thrm]{Corollary}
\theoremstyle{definition}

\numberwithin{equation}{section}

\newcounter{questionnumber}

\newcommand{\bm}[1]{\mbox{\boldmath{$#1$}}}
\begin{document}

\title{\large A lower bound on the fidelity between two states in terms of their Bures distance}

\author{Lin Zhang\footnote{E-mail: godyalin@163.com; linyz@zju.edu.cn}\\
  {\it\small Institute of Mathematics, Hangzhou Dianzi University, Hangzhou 310018, PR~China}\\
  Junde Wu\\
  {\it\small Department of Mathematics, Zhejiang University, Hangzhou 310027,
  PR~China}}
\date{}
\maketitle
\maketitle \mbox{}\hrule\mbox\\
\begin{abstract}

Fidelity is a fundamental and ubiquitous concept in quantum
information theory. Fuchs-van de Graaf's inequalities deal with
bounding fidelity from above and below. In this paper, we give a
lower bound on the quantum fidelity between two states in terms of
their Bures distance.
\end{abstract}
\maketitle \mbox{}\hrule\mbox\\

\section{Introduction}

The \emph{fidelity} between two quantum states, represented by
density operators $\rho$ and $\sigma$, is defined as
\begin{eqnarray*}
\rF(\rho, \sigma ) = \Tr{\sqrt{\sqrt{\rho}\sigma\sqrt{\rho}}}.
\end{eqnarray*}
Note that both density operators here are taken from
$\density{\cH_d}$, the set of all positive semi-definite operator
with unit trace on a $d$-dimensional Hilbert space $\cH_d$. The
squared fidelity above has been called \emph{transition probability}
\cite{Uhlmann76,Uhlmann2011}. Operationally it is the maximal
success probability of changing a state to another one by a
measurement in a larger quantum system. The fidelity is also
employed in a number of problems such as quantifying entanglement
\cite{Vedral}, and quantum error correction \cite{Kosut}, etc.

For quantum fidelity, the well-known Fuchs-van de Graaf's inequality
states that: For arbitrary two density operators $\rho$ and $\sigma$
in $\density{\cH_d}$, it holds that
\begin{eqnarray}
1-\frac12\norm{\rho-\sigma}_1 \leqslant \rF(\rho, \sigma) \leqslant
\sqrt{1-\frac14\norm{\rho-\sigma}^2_1},
\end{eqnarray}
which established a close relationship between the trace-norm of the
difference for two density operators and their fidelity
\cite{Watrous08}.

Given two density operators $\rho,\sigma\in\density{\cH_d}$,
$\set{\lambda\rho + (1-\lambda)\sigma: \lambda\in [0,1]}$ is the
affine mixed path of $\rho$ and $\sigma$. In the recent papers
\cite{Audenaert1,Audenaert2,Audenaert3}, the author considered the
estimate of relative entropy between a fixed state $\rho$ and the
affine mixed path $\lambda\rho + (1-\lambda)\sigma$. In this paper,
we will consider similar problem for fidelity.

\section{Main result}

The Fuchs-van de Graaf's inequality can not be improved because it
is tight. For any value of the trace-norm between $\rho$ and
$\sigma$, there exists a pair of states saturating the inequality.
However, by supplying additional information about the pair it is
possible to obtain a higher lower bound on the fidelity.

\begin{thrm}\label{th:zhang}
Let $\rho$ and $\sigma$ be two pure states in $\density{\cH_d}$,
$\lambda\in [0,1]$. Then
\begin{eqnarray}\label{eq:improved-ineq}
\rF(\rho, \lambda\rho + (1-\lambda)\sigma) \geqslant
1-\frac12(1-\sqrt{\lambda})\norm{\rho-\sigma}_1.
\end{eqnarray}
\end{thrm}

\begin{proof}
Let $\rho = \out{\psi}{\psi}$ and $\sigma=\out{\phi}{\phi}$ for
nomalized vectors $\ket{\psi},\ket{\phi}\in\cH_d$. Without loss of
generality, assume that $\lambda\in(0,1)$. Denote
$\abs{\iinner{\psi}{\phi}} := r$. Then $r\in [0,1]$. Now
\begin{eqnarray}
\rF(\rho, \lambda\rho + (1-\lambda)\sigma) &=&
\sqrt{\Innerm{\psi}{\lambda\out{\psi}{\psi} +
(1-\lambda)\out{\phi}{\phi}}{\psi}}\\
&=& \sqrt{\lambda + (1-\lambda)r^2}.
\end{eqnarray}
Note that
\begin{eqnarray}
\norm{\out{\psi}{\psi}-\out{\phi}{\phi}}_1 =
2\sqrt{1-\abs{\iinner{\psi}{\phi}}^2}.
\end{eqnarray}
This implies that
\begin{eqnarray}
\norm{\rho-\sigma}_1 = 2\sqrt{1-r^2}.
\end{eqnarray}
Next, all we have to do is to show that
\begin{eqnarray}
\sqrt{\lambda + (1-\lambda)r^2} \geqslant
1-(1-\sqrt{\lambda})\sqrt{1-r^2}.
\end{eqnarray}
Denote by $f(\lambda,r)$ the difference of the squared both sides as
follows:
\begin{eqnarray}
f(\lambda,r) &:=& \Pa{\sqrt{\lambda + (1-\lambda)r^2}}^2 -
\Pa{1-(1-\sqrt{\lambda})\sqrt{1-r^2}}^2\\
&=& 2(1-\sqrt{\lambda})\Pa{\sqrt{1-r^2}-(1-r^2)}
\end{eqnarray}
Since $r\in[0,1]$, i.e. $1-r^2\in[0,1]$, it follows that
$\sqrt{1-r^2}\geqslant1-r^2$. Therefore $f(\lambda,r)\geqslant0$. We
can conclude that
\begin{eqnarray}
\sqrt{\lambda + (1-\lambda)r^2} \geqslant
1-(1-\sqrt{\lambda})\sqrt{1-r^2}.
\end{eqnarray}
The desired conclusion is obtained.
\end{proof}

The Bures distance is a very useful quantity in quantum information
theory, the Bures distance between two states is defined as
$\rB(\rho,\sigma):=\sqrt{1-\rF(\rho,\sigma)^2}$. It is seen easily
that $\frac12\norm{\rho-\sigma}_1\leqslant\rB(\rho,\sigma)$ by the
Fuchs-van de Graaf's inequality. We have the following weaker
inequality like the \emph{conjectured} inequality \eqref{eq:conj}:
\begin{cor}
Let $\rho$ and $\sigma$ be two density operators in
$\density{\cH_d}$, $\lambda\in [0,1]$. Then
\begin{eqnarray}\label{eq:improved-ineq}
\rF(\rho, \lambda\rho + (1-\lambda)\sigma) \geqslant 1-
(1-\sqrt{\lambda})\rB(\rho,\sigma).
\end{eqnarray}
\end{cor}

\begin{proof}
Let $\ket{\Psi}$ and $\ket{\Phi}$ be any purifications of $\rho$ and
$\sigma$, respectively such that $\rF(\rho,\sigma) =
\abs{\iinner{\Psi}{\Phi}}$. Then
\begin{eqnarray}
\rF(\rho, \lambda\rho + (1-\lambda)\sigma) &\geqslant&
\rF(\out{\Psi}{\Psi},\lambda\out{\Psi}{\Psi}+(1-\lambda)\out{\Phi}{\Phi})\\
&\geqslant& 1 - (1-\sqrt{\lambda})\frac12\norm{\out{\Psi}{\Psi} -
\out{\Phi}{\Phi}}_1\\
&=&1 - (1-\sqrt{\lambda})\sqrt{1-\abs{\iinner{\Psi}{\Phi}}^2},
\end{eqnarray}
implying the desired inequality.
\end{proof}

\section{Discussion}

Naturally, we have the following \emph{conjecture}: Let $\rho$ and
$\sigma$ be \emph{any} two states in $\density{\cH_d}$, $\lambda\in
[0,1]$. Then
\begin{eqnarray}\label{eq:conj}
\rF(\rho, \lambda\rho + (1-\lambda)\sigma) \geqslant
1-\frac12(1-\sqrt{\lambda})\norm{\rho-\sigma}_1.
\end{eqnarray}
If this conjecture \emph{were} true, then we \emph{would} supply the
max-relative entropy between the states as additional information in
the lower bound of fidelity. The max-relative entropy is defined as
$$
\rS_{\max}(\rho||\sigma) :=\inf_\gamma\set{\gamma: \rho\leqslant
e^\gamma\sigma}.
$$
Clearly $e^{\rS_{\max}(\rho||\sigma)} =
\lambda_{\max}(\sigma^{-1/2}\rho\sigma^{-1/2})$. The inequality
\eqref{eq:conj} is about the special pair of states $\rho$ and
$\lambda\rho+(1-\lambda)\sigma$ and seems to be of rather restricted
importance. However it is possible to reformulate the inequality as
an inequality about any pair of states, we will now show the
following conclusion:
\begin{eqnarray}\label{eq:lower-bound-2}
\rF(\rho,\sigma)\geqslant
1-\frac12\frac{\sqrt{e^{\rS_{\max}(\rho||\sigma)}}}{\sqrt{e^{\rS_{\max}(\rho||\sigma)}}+1}\norm{\rho
- \sigma}_1
\end{eqnarray}
for $\rho,\sigma\in\density{\cH_d}$. Indeed, given any two density
operators $\rho$ and $\sigma$, we know that if $\rho\sigma\neq0$,
then
$$
\min\set{\lambda>0: \rho\leqslant \lambda\sigma} =
\lambda_{\max}(\sigma^{-1/2}\rho\sigma^{-1/2}) := \lambda_0,
$$
where the notation $\lambda_{\max}(X)$ is used to denote the maximum
eigenvalue of the operator $X$, we define also $\min\set{\lambda>0:
\rho\leqslant \lambda\sigma} = +\infty$ if $\rho\sigma=0$. Clearly
$\lambda_0>0$. If denote $\hat\sigma:= \frac{\sigma -
\lambda^{-1}_0\rho}{1-\lambda^{-1}_0}$, then
$$
\sigma =  \lambda^{-1}_0\rho + (1- \lambda^{-1}_0)\hat\sigma.
$$
Therefore
\begin{eqnarray*}
\rF(\rho,\sigma) = \rF(\rho,\lambda^{-1}_0\rho + (1-
\lambda^{-1}_0)\hat\sigma) \geqslant
1-\frac12\Pa{1-\sqrt{\lambda^{-1}_0}}\norm{\rho-\sigma}_1
\end{eqnarray*}
implies that
$$
\rF(\rho,\sigma)\geqslant
1-\frac12\frac{\sqrt{\lambda_0}}{\sqrt{\lambda_0}+1}\norm{\rho-\sigma}_1~~\text{for}~~\lambda_0>0.
$$
The lower bound in the above inequality is indeed tighter than one
in Fuchs-van de Graaf's inequality. Thus, we get a state-dependent
factor in the lower bound for fidelity. That is, when $\lambda_0=
+\infty$, the above lower bound is reduced to the lower bound in
Fuchs-van de Graaf's inequality.

For the related problems along this line such as min- and max-
(relative) entropy, it is referred to \cite{Datta}.

\section{Conclusion}

In this paper, we obtained a lower bound on the fidelity between a
fixed state and its a mixed path with another state. Based on this
result, we derived a lower bound  on the fidelity between two
states. The result may shed new light on some related problems in
quantum information theory, for example, ones can consider bounding
the input-output fidelity of transpose channel.

\subsubsection*{Acknowledgements}
LZ and JW are supported by the Natural Science Foundations of China
(11301124, 11171301) and the Doctoral Programs Foundation of
Ministry of Education of China (J20130061).


\end{document}